\documentclass[oneside,english]{amsart}
\usepackage[T1]{fontenc}
\usepackage[latin9]{inputenc}
\usepackage{amsthm}

\makeatletter
\numberwithin{equation}{section}
\numberwithin{figure}{section}
  \theoremstyle{plain}
  \newtheorem*{lem*}{Lemma}
\theoremstyle{plain}
\newtheorem{thm}{Theorem}
  \theoremstyle{plain}
  \newtheorem*{cor*}{Corollary}
  \theoremstyle{plain}
  \newtheorem*{thm*}{Theorem}

\makeatother

\usepackage{babel}

\begin{document}

\title{Differential geometric formulation of the Cauchy Navier equations }

\date{December 15th, 2009}

\author{Frank Schadt, Pforzheim University of Applied Sciences, Germany}

\maketitle
\begin{abstract}
The paper presents a reformulation of some of the most basic entities
and equations of linear elasticity \textendash{} the stress and strain
tensor, the Cauchy Navier equilibrium equations, material equations
for linear isotropic bodies - in a modern differential geometric language
using differential forms and lie derivatives. Similar steps have been
done successfully in general relativity, quantum physics and electrodynamics
and are of great use in those fields. In Elasticity Theory, however,
such a modern differential geometric approach is much less common.
Furthermore, existing reformulations demand a vast knowledge of differential
geometry, including nonstandard entities such as vector valued differential
forms and the like. This paper presents a less general but more easily
accessible approach to using modern differential geometry in elasticity
theory than those published up to now. 
\end{abstract}

\section{Introduction}

Modern differential geometry often allows for a clearer, more geometric
approach to solve physical problems than Gibbsian vector calculus
or Ricci calculus (i.e. the index-based form) of Riemannian geometry.
Even more importantly, equations stated in modern differential geometric
terms are not bound to the use of special classes of coordinate systems,
such as cartesian or orthonormal ones. It provides useful and powerful
concepts such as differential forms and Lie derivatives and with them
a generalization of vector calculus expressions like rotation, gradient
and divergence and of various integral theorems like Gauss\textquoteright{}
divergence theorem or Stoke\textquoteright{}s integral theorem on
planes to just one integral theorem. Furthermore, coordinate transformations
in curvilinear coordi-nate systems are often simplified. 

In general, it is desirable to express equations in a connection and
metric ten-sor free manner, using only differential forms and Lie
derivatives, thus simplifying coordinate system changes. Unfortunately
it is not always possible to express vector or tensor equations as
equations in such a way. One reason of this is the antisymmetric nature
of differential forms of degree \ensuremath{\ge}2. While e.g. the
gradient of a scalar function and the rotation and divergence of a
vector field can be expressed using differential forms, this is not
possible with the gradient of a vector field . It is, however, possible
to reformulate certain differential equations of vector fields as
ones of differential forms.

In electrodynamics, differential forms are already a quite popular
way to reformulate maxwell\textquoteright{}s equations etc. \cite{Selfridge}
In continuum mechanics, however, there are only few steps towards
using differential forms, most notably from J. E. Marsden \cite{Marsden}
and V.I. Arnold \cite{Arnold}. In \cite{Kanso}, the stress and strain
tensors are introduced as vector or covector valued differential forms
of various degree. In that approach, the physical meaning of those
tensors are clearer than in the traditional way of expressing them
interchangeably as contravariant, covariant or mixed tensors. However,
vector or covector valued differential forms are a nonstandard subject
in differential geometry and usually do not simplify calculations. 

For that reason, in this paper a path between traditional tensor analysis
and the very general and mathematically rigorous approach of \cite{Marsden}
is taken, thereby restricting to linear elasticity. Due to their relevance,
only static problems without body forces are considered. 

Absolutely necessary to determine deformations inside a body is a
suitable system of differential equations for displacements, stresses
or strains, respectively. Since displacement fields can be considered
as fundamental in the sense that stress or strain fields can be determined
from them, the cauchy-navier displacement equations shall be reformulated
in chapter\eqref{sec:Displacements}. The resulting set of equations
have only the displacement variables as unknowns, so displacement
boundary conditions are easy to handle in this system of equations.
In order to allow also traction boundary conditions and to relate
stresses, strains and displacements, it is necessary to express the
well-known stress and strain tensor of linear elasticity in the same
modern differential geometric formalism . This will be done in chapter
\eqref{sec:Traction}. 

It shall be noted that throughout the paper it was assumed that manifolds
have sufficiently smooth boundaries and boundary values are also sufficiently
smooth.

\section{\label{sec:Displacements}Cauchy Navier\textquoteright{}s Displacement
Equations of Static Linear Elasticity}

\subsection{Basic Definitions}

Problems in static linear elasticity consist of finding either stress
fields that satisfy the equilibrium equation $\nabla\cdot\overleftrightarrow{\sigma}=0$
or, equivalently, displacements that fulfil the Cauchy Navier equation\cite{Timoshenko}
inside a suitable submanifold $B$ of $E^{3}$: \begin{equation}
\mu\triangle\vec{u}+(\lambda+\mu)\nabla(\nabla\cdot\vec{u})=0\end{equation}

In this article, we will regard the latter equation as fundamental
because stresses and strains can easily be calculated from displacements.
Also, there is no need for additional compatibility equations to be
satisfied as is the case for the equilibrium equation for stresses
or strains.

Displacements in elasticity theory are naturally tangent vector fields,
e.g. $\vec{u}:E^{3}\supseteq B\rightarrow TB$. In order to make use
of exterior calculus, a corresponding displacement covector field
$u:E^{3}\rightarrow T^{*}E^{3}$ is introduced, which is derived trivially
from the displacement vector field by index lowering: 

\begin{equation}
u:=\vec{u}^{b}=g(\vec{u})\label{eq:0}\end{equation}

Here, $g(\vec{u})$ is the action of the covariant metric tensor on
the vector $\vec{u}$. The superscript b indicates the index lowering
process or association of a covector field to a vector field, respectively.
The inverse action is the index raising operation $\vec{u}=u^{\#}=g^{-1}(u)$.

\subsection{Vector Differential Operators in Exterior Calculus }
\begin{lem*}
For the differential operators grad, div and rot, applied on vector
or scalar fields, respectively, the following relationships are valid:
\end{lem*}
\begin{equation}
\nabla f=df^{\#}\label{eq:1}\end{equation}

\begin{equation}
\nabla\cdot\overrightarrow{u}=*d*u\label{eq:2}\end{equation}

\begin{equation}
\nabla\times\overrightarrow{u}=(*du)^{\#}\label{eq:3}\end{equation}

Proofs for these relations can be found for example in \cite{Flanders}.

\subsection{Cauchy Navier Equation}
\begin{thm}
Given a displacement vector field $\vec{u}$ that satisfies the Cauchy
Navier equations, then the corresponding covector field $u$ satisfies 
\end{thm}
\begin{equation}
(\lambda+2\mu)d\delta u+\mu\delta du=0\label{eq:4}\end{equation}

\begin{proof}
The correspondence between the divergence or rotation of vector fields
and exterior derivatives of differential forms is well known. Since,
however, there is no equivalent of the gradient of a vector field
in exterior calculus, the Laplace operator has to be restated using
the well known formula

\begin{equation}
\triangle\vec{u}=\nabla\cdot\nabla\vec{u}=\nabla(\nabla\cdot\vec{u})-\nabla\times\nabla\times\vec{u}\end{equation}

With this, the Cauchy Navier equation appears in the equivalent form 

\begin{equation}
(\lambda+2\mu)\nabla(\nabla\cdot\vec{u})-\mu\nabla\times\nabla\times\vec{u}=0\label{eq:5}\end{equation}

Now, using \ref{eq:1} to \ref{eq:3}, we can restate \ref{eq:5}
in terms of the corresponding displacement one-form $u=\vec{u}^{b}$:

\begin{equation}
(\lambda+2\mu)d*d*u-\mu*d*du=0\label{eq:6}\end{equation}

A somewhat more concise notation of \ref{eq:6} uses the codifferential
operator, defined in $E^{3}$ as

\begin{eqnarray}
\delta:\Omega^{n}(E^{3}) & \rightarrow & \Omega^{n-1}(E^{3})\\
\omega & \mapsto & \begin{array}{cl}
-*d*\omega & \; if\,\omega\in\Omega^{1}(E^{3})\: or\:\omega\in\Omega^{3}(E^{3})\\
*d*\omega & \; if\,\omega\in\Omega^{2}(E^{3})\end{array}\nonumber \end{eqnarray}

With this we arrive at the proposed formula

\begin{equation}
(\lambda+2\mu)d\delta u+\mu\delta du=0\end{equation}

\end{proof}

\subsection{Boundary Conditions }

The treatment of displacement boundary conditions are trivial: the
displacement vector field on the boundary just has to be mapped to
its corresponding covector field. 
\begin{cor*}
Given a displacement vector field $\vec{u}$ that satisfies the Cauchy
Navier equations and assumes given values $\vec{u}|\partial B$ on
the boundary $\partial B$, then the corresponding covector field
$u$ satisfying \ref{eq:4} assumes the values $g(\vec{u})$ on the
same boundary.
\end{cor*}
Traction boundary conditions, however, are much harder to handle.
It would be desirable to find a way to express traction forces or
the elastic stress tensor in terms of the displacement covector field,
resulting in boundary conditions that can be stated as a set of differential
equations in the discplacement covector field. This will be the content
of the next and largest part of this paper.

\section{\label{sec:Traction}Traction Boundary Conditions }

\subsection{The Linear Strain Tensor }

According to \cite{Marsden}, the strain tensor $\overleftrightarrow{\epsilon}\in T_{2}^{0}(E^{3})$
of linear elasticity $\overleftrightarrow{\epsilon}=\epsilon_{ij}e^{i}\otimes e^{j}=\frac{1}{2}(\nabla\vec{u}+\nabla\vec{u}^{T})$
can also be written as the Lie derivative of the metric tensor with
respect to the displacement vector field $\vec{u}$:

\begin{equation}
\overleftrightarrow{\epsilon}=\frac{1}{2}L_{\vec{u}}\overleftrightarrow{g}\label{eq:7}\end{equation}

with the metric tensor $\overleftrightarrow{g}=g_{ij}dx^{i}\otimes dx^{j}$.
Beware that we need the discplacement vector field here, not the corresponding
one-form. This is unfortunate, because we now need to map one-form
into vector field components and vice versa. However, in the case
of orthonormal bases and physical vector components this is simple:
the components are unchanged. Otherwise, index lowering or raising
has to be used. 

As an example, using a general canonical vector base, it shall be
shown, that the formula above is equivalent to results of classical
tensor analysis. 
\begin{cor*}
The components of \ref{eq:7} satisfy 
\end{cor*}
\begin{equation}
\overleftrightarrow{\epsilon}=\epsilon_{ij}dx^{i}\otimes dx^{j}=\frac{1}{2}(u^{k}g_{ij,k}+g_{kj}u_{,i}^{k}+g_{ik}u_{,j}^{k})dx^{i}\otimes dx^{j}\label{eq:strain}\end{equation}

\begin{proof}
This can be proven by direct calculation. 

The Lie derivative adheres to the product rule for tensor products:\begin{equation}
L_{\vec{x}}(T\otimes S)=L_{\vec{x}}T\otimes S+T\otimes L_{\vec{x}}S\end{equation}

This holds for tensor fields of any rank, and therefore especially
for scalar functions and one-forms. Thus we have 

\begin{equation}
L_{\vec{u}}(g_{ij}dx^{i}\otimes dx^{j})=(L_{\vec{u}}g_{ij})dx^{i}\otimes dx^{j}+g_{ij}L_{\vec{u}}dx^{i}\otimes dx^{j}+g_{ij}dx^{i}\otimes L_{\vec{u}}dx^{j}\label{eq:12}\end{equation}

Cartan\textquoteright{}s \quotedblbase{}magic formula\textquotedblleft{}
can be applied to calculate the lie derivatives of one-forms:

\begin{equation}
L_{\vec{x}}\alpha=\vec{x}\cdot d\alpha+d(\vec{x}\cdot\alpha)\end{equation}

Applied to the canonical basis covectors $dx^{i}$and afterwards to
\ref{eq:12}:

\begin{equation}
L_{\vec{u}}dx^{i}=du^{i}=u_{,k}^{i}dx^{k}\end{equation}

\begin{eqnarray}
L_{\vec{u}}(g_{ij}dx^{i}\otimes dx^{j}) & = & u^{k}g_{ij,k}dx^{i}\otimes dx^{j}+g_{ij}u_{,k}^{i}dx^{k}\otimes dx^{j}+g_{ij}u_{,k}^{j}dx^{i}\otimes dx^{k}\nonumber \\
 & = & (u^{k}g_{ij,k}+g_{kj}u_{,i}^{k}+g_{ik}u_{,j}^{k})dx^{i}\otimes dx^{j}\end{eqnarray}

Thus we have \[
\overleftrightarrow{\epsilon}=\epsilon_{ij}dx^{i}\otimes dx^{j}=\frac{1}{2}(u^{k}g_{ij,k}+g_{kj}u_{,i}^{k}+g_{ik}u_{,j}^{k})dx^{i}\otimes dx^{j}\]
\end{proof}
\begin{cor*}
The components of the linear elastic strain tensor defined as $\tilde{\epsilon}_{ij}=\frac{1}{2}(u_{i|j}+u_{j|i})$
satisfy \ref{eq:strain}.\end{cor*}
\begin{proof}
Let's first rewrite $\tilde{\epsilon}_{ij}=\frac{1}{2}(u_{i|j}+u_{j|i})$
involving the Christoffel symbols of the second kind:

\begin{eqnarray}
\tilde{\epsilon}_{ij} & = & \frac{1}{2}(u_{i,j}-\Gamma_{ij}^{k}u_{k}+u_{j,i}-\Gamma_{ji}^{k}u_{k})\nonumber \\
 & = & \frac{1}{2}(u_{i,j}+u_{j,i})-\Gamma_{ij}^{k}u_{k}\end{eqnarray}

In the last step we used the symmetry property of the Christoffel
symbol $\Gamma_{ij}^{k}=\Gamma_{ji}^{k}$.

Now we raise the index of the $u_{i}$components and substitute $\Gamma_{ij}^{k}=\frac{1}{2}g^{km}(g_{im,j}+g_{jm,i}-g_{ij,m})$:

\begin{eqnarray}
\tilde{\epsilon}_{ij} & = & \frac{1}{2}((g_{ik}u^{k})_{,j}+(g_{jk}u^{k})_{,i})-\frac{1}{2}g^{km}(g_{im,j}+g_{jm,i}-g_{ij,m})u_{k}\nonumber \\
 & = & \frac{1}{2}(g_{ik,j}u^{k}+g_{ik}u_{,j}^{k}+g_{jk,i}u^{k}+g_{jk}u_{,i}^{k}-g_{im,j}u^{m}-g_{jm,i}u^{m}+g_{ij,m}u^{m})\nonumber \\
 & = & \frac{1}{2}(g_{ik,j}u^{k}+g_{ik}u_{,j}^{k}+g_{jk,i}u^{k}+g_{jk}u_{,i}^{k}-g_{ik,j}u^{k}-g_{jk,i}u^{k}+g_{ij,k}u^{k})\nonumber \\
 & = & \frac{1}{2}(g_{ik}u_{,j}^{k}+g_{jk}u_{,i}^{k}+g_{ij,k}u^{k})\end{eqnarray}

A comparison shows that $\tilde{\epsilon}_{ij}=\epsilon_{ij}$.
\end{proof}

\subsection{The Stress Tensor}

For isotropic homogeneous bodies, the stress and strain tensor of
linear elasticity are related as follows \cite{Marsden}: 

\begin{equation}
\overleftrightarrow{\sigma}=\lambda e\overleftrightarrow{g}+2\mu\overleftrightarrow{\epsilon}\end{equation}

Here, $e=\nabla\cdot\vec{u}$ is the volume expansion under deformation.
For a displacement one-form $u$ we have $e=-\delta u=*d*u$. Since
$\overleftrightarrow{\epsilon}=\frac{1}{2}L_{\vec{u}}\overleftrightarrow{g}$
we get: 

\begin{equation}
\overleftrightarrow{\sigma}=-\lambda\delta u\overleftrightarrow{g}+\mu L_{\vec{u}}\overleftrightarrow{g}\end{equation}

Beware that as before, unfortunately, we have to deal with the displacement
one-form and vector field at the same time.

\subsection{Traction Boundary Conditions }
\begin{lem*}
Given a traction vector field $\vec{t}:\partial B\supset\partial S\rightarrow TB$
on a part of $B$s boundary. The corresponding traction covector field
$t:=\vec{t}^{b}$ then satisfies
\end{lem*}
\begin{equation}
t:=\vec{t}^{b}=\overleftrightarrow{\sigma}\cdot\vec{n}=\lambda en+2\mu\overleftrightarrow{\epsilon}\cdot\vec{n}\label{eq:23}\end{equation}

\begin{proof}
This is a trivial implication of cauchy's stress theorem\end{proof}
\begin{thm*}
Given a traction covector field $t:\partial B\supset\partial S\rightarrow T^{*}B$
and a displacement covector $u$ field satisfying \ref{eq:4}. The
thus given traction boundary conditions are satisfied, if the displacement
covector field and its corresponding vector field $\vec{u}$ adheres
also to the following equation on \textup{$\partial S$:}\\
\begin{equation}
t=-\lambda\delta un+\mu(d(\vec{u}\cdot n)+\vec{u}\cdot dn+[\vec{n},\vec{u}]^{b})\label{eq:traction}\end{equation}
\\
where\textup{ $\vec{n}$} is a normal vector at a given point
and $n=\vec{n}^{b}$ its corresponding one-form.\end{thm*}
\begin{proof}
Since the Lie derivative is also a derivative with respect to scalar
products ($L_{\vec{x}}(T\cdot S)=L_{\vec{x}}T\cdot S+T\cdot L_{\vec{x}}S$
for any two tensors $T$ and $S$) we can write:

\begin{eqnarray*}
2\overleftrightarrow{\epsilon}\cdot\vec{n} & = & L_{\vec{u}}\overleftrightarrow{g}\cdot\vec{n}=L_{\vec{u}}(\overleftrightarrow{g}\cdot\vec{n})-\overleftrightarrow{g}\cdot L_{\vec{u}}\vec{n}\\
 & = & L_{\vec{u}}\vec{n}-[\vec{u},\vec{n}]^{b}=d(\vec{u}\cdot n)+\vec{u}\cdot dn+[\vec{n},\vec{u}]^{b}\end{eqnarray*}

Now only the part $\lambda en$ is missing from \ref{eq:23}, which
is equal to $-\lambda\delta un$.
\end{proof}
Now let\textquoteright{}s assume we have a coordinate system on which
one of the canonical base vectors equals the surface normal $\vec{n}$.
In accordance with cylindrical or spherical coordinates, let\textquoteright{}s
call this coordinate function $r$, the canonical base vector along
this coordinate $\partial_{r}$ and the canonical base covector $dr$.
Let\textquoteright{}s also assume that $\vec{n}$ (and thus $n$)
are normalized, as well as $\partial_{r}$ and $dr$. In this case
$\vec{n}=\partial_{r}$ and $n=dr$. Then we can reformulate \ref{eq:traction}
according to the following
\begin{cor*}
Given a coordinate system with a coordinate function $r$ satisfying
the description above, a traction covector field $t:\partial B\supset\partial S\rightarrow T^{*}B$
and a displacement covector $u$ field satisfying \ref{eq:4}, the
latter have to adhere to the following equation
\end{cor*}
\begin{eqnarray}
t & = & -(\lambda\delta u)dr+\mu(du^{r}+[\partial_{r},\vec{u}]^{b})\label{eq:26}\end{eqnarray}

\begin{proof}
We look at all the substitions one after the other. $-\lambda\delta un=-(\lambda\delta u)dr$
just by the definition $n=dr$ given above.

$d(\vec{u}\cdot n)=du^{r}$ follows from the fact that $\vec{u}\cdot n=\vec{u}\cdot dr=u^{r}$. 

$\vec{u}\cdot dn$ vanishes because $dn=ddr=0$.

$[\vec{n},\vec{u}]^{b}=[\partial_{r},\vec{u}]^{b}$ also by definition.
\end{proof}

\section{Conclusion }

In the first part of this article we have derived Cauchy Navier\textquoteright{}s
equation of static equilibrium in linear elasticity using only differential
forms instead of vector fields. The second part treats the harder
problem of expressing strain and stress tensors in a more modern differential
geometric way, and, finally, how to deal with traction boundary conditions.

\end{document}